\RequirePackage{amsmath}
\documentclass[runningheads]{llncs}

\usepackage[final]{microtype}

\usepackage{graphics}
\usepackage{subcaption}
\usepackage{amssymb}

\usepackage{booktabs}
\usepackage{algorithm}
\usepackage{algorithmic}
\usepackage{tikz}

\usepackage{hyperref}

\tikzset{node/.style={circle,fill=gray!10,draw,minimum size=0.6cm,inner sep=0pt} }
\tikzset{arc/.style = {->,> = latex, thick, } }

\renewcommand{\P}[1]{\mathbb{P}\left(#1\right)}
\newcommand{\E}[1]{\mathbb{E}\left[#1\right]}
\renewcommand{\setminus}{-}
\newcommand{\algname}{\textsc{Dispatch}}
\newcommand{\tpp}{\textsc{TPP}}

\title{DISPATCH: An Optimally-Competitive Algorithm for Maximum Online Perfect Bipartite Matching with i.i.d.~Arrivals}

\titlerunning{DISPATCH: Online Perfect Bipartite Matching with i.i.d.~Arrivals}

\author{Minjun Chang\inst{1}\orcidID{0000-0002-2251-4163}
\and Dorit S. Hochbaum\inst{1}\orcidID{0000-0002-2498-0512}
\and \\Quico Spaen\inst{1}\orcidID{0000-0003-2788-1904}
\and Mark Velednitsky\inst{1}\orcidID{0000-0003-1176-5159}}

\institute{University of California, Berkeley\\ \email{\{minjun.lynn,dhochbaum,qspaen,marvel\}@berkeley.edu}}

\begin{document}

\maketitle

\begin{abstract}
This work presents an optimally-competitive algorithm for the problem of maximum weighted online perfect bipartite matching with i.i.d.~arrivals. In this problem, we are given a known set of workers, a distribution over job types, and non-negative utility weights for each pair of worker and job types. At each time step, a job is drawn i.i.d.~from the distribution over job types. Upon arrival, the job must be irrevocably assigned to a worker and cannot be dropped. The goal is to maximize the expected sum of utilities after all jobs are assigned.

We introduce \algname{}, a 0.5-competitive, randomized algorithm. We also prove that 0.5-competitive is the best possible. \algname{} first selects a ``preferred worker'' and assigns the job to this worker if it is available. The preferred worker is determined based on an optimal solution to a fractional transportation problem. If the preferred worker is not available, \algname{} randomly selects a worker from the available workers. We show that \algname{} maintains a uniform distribution over the workers even when the distribution over the job types is non-uniform.


\end{abstract}

\section{Introduction}	
We consider the problem of \emph{maximum online perfect bipartite matching}. Suppose that we have a set of jobs and a set of workers. At every time step, a single job arrives to be served by one of the workers. Upon a job's arrival, we observe the utility of assigning the job to each of the workers. We must immediately decide which worker will serve the job. Once a worker is assigned a job, it is busy and cannot be assigned to another job. Jobs continue to arrive until all workers are busy.

In the natural bipartite graph that arises, there is an edge between each worker and job with a non-negative utility of assigning that worker to that job. The assignment of workers to jobs will form a perfect matching in this bipartite graph. Our goal is to design a dispatching algorithm that maximizes the expected sum of utilities of the perfect matching. 

In this work, we consider the maximum online perfect bipartite matching problem with \emph{independent and identically distributed (i.i.d.) arrivals}. This means that, at each time step, a job is drawn i.i.d.~from a known distribution over job types.

Examples of online bipartite matching include matching doctors to patients in hospitals, matching operators to callers in call centers, matching drivers to passengers in ride-sharing, and matching impressions to customers in online ad auctions~\cite{Meh13}.

We introduce the randomized algorithm \algname{} for the problem of online weighted perfect bipartite matching with i.i.d.~arrivals. \algname{} is 0.5-competitive algorithm: the total expected utility of the perfect matching produced by \algname{} is at least half of the total expected utility of an optimal algorithm that knows the job arrival sequence in advance. We also describe a family of problem instances for which 0.5 is the best-possible competitive ratio. The \algname{} algorithm, thus, achieves the best-possible competitive ratio. In contrast, the same problem with adversarial job arrivals cannot be bounded, as observed by Feldman et al.~\cite{feldman2009wine}.

To assign workers to jobs, \algname{} first selects a \emph{preferred worker}. This preferred worker is determined based on an optimal solution to a fractional transportation problem. If the preferred worker is available, then job is assigned to this worker. Otherwise, \algname{} randomly selects a worker from the available workers.

% We show that \algname{} maintains a uniform distribution over the workers even though the distribution over the job types may be non-uniform. In addition, we prove that no algorithm can attain a competitive ratio better than $\frac{1/2}{1-1/e}\approx 0.79$. We will define the competitive ratio formally in section \ref{sec:prelims}.

\subsection{Related Work}
Our work resides in the space of online matching problems, including the Maximum (Imperfect) Bipartite Matching problem and the Minimum (Perfect) Bipartite Matching problem. Another closely related problem is the $k$-Server problem. For each of these problems, several arrival models are considered. Arrival models including adversarial, where the adversary chooses jobs and their arrival order; random order, where the adversary chooses jobs but not their arrival order; and i.i.d., where the adversary specifies a probability distribution over job types and each arrival is sampled independently from the distribution. We briefly describe each of these problems and present best-known results, contrasting it to the setting considered here. A summary is in Table \ref{tab:literature}.

% Reassigning only makes sense when it changes the underlying graph.

\subsubsection{Maximum Online (Imperfect) Bipartite Matching}
The maximum online (imperfect) bipartite matching problem is defined on a bipartite graph with $n$ known workers and $n$ jobs that arrive one at a time. Jobs either get assigned to a worker or are discarded. The goal is to maximize the cardinality (or sum of weights) of the resulting matching. In contrast to our problem, jobs may be the discarded and the resulting matching may be \emph{imperfect}.

For the unweighted problem with adversarial arrivals, Karp, Vazirani, and Vazirani~\cite{KarVazVaz90} showed a best-possible algorithm that achieves a competitive ratio of $1-\frac{1}{e} \approx 0.632$. Variations of the problem have been proposed: addition of edge or vertex weights, the use of budgets, different arrival models, etc. Mehta~\cite{Meh13} provides an excellent overview of this literature. When the arrivals are in a random order, it is possible to do better than $1-\frac{1}{e}$. Mahdian and Yan~\cite{mahdian2011online}, in $2011$, achieved a competitive ratio of $0.696$. Manshadi et al.~\cite{ManGhaSab12} showed that you cannot do better than $0.823$. If the problem also has weights, then the best-possible competitive ratio is $0.368$ by a reduction from the secretary problem as shown by Kesselheim et al.~\cite{kesselheim2013optimal}. They also give an algorithm that attains this competitive ratio. 

The problem has also been studied when the jobs are drawn i.i.d.~from a known distribution. This problem is also referred to as \emph{Online Stochastic Matching}. The first result to break the $1 - \frac{1}{e}$ barrier for the unweighted case was the $0.67$-competitive algorithm of Feldman et. al.~\cite{feldman2009online} in $2009$. To date, the best-known competitive ratio of $0.730$ is due to Brubach et al~\cite{BruSanSri16}. This is close the best-known bound of 0.745 by Correa et al.~\cite{correa2017posted}.

\subsubsection{Online Minimum (Perfect) Bipartite Matching}
The online minimum (perfect) bipartite matching addresses the question of finding a minimum cost perfect matching on a bipartite graph with $n$ workers and $n$ jobs. Given any arbitrary sequence of jobs arriving one by one, each job needs to be irrevocably assigned to worker on arrival. This problem is the minimization version of the problem considered in this work. However, the obtained competitive ratios do not transfer.

The problem was first considered by Khuller, Mitchell, and Vazirani~\cite{KhuMitVaz94} and independently by Kalyanasundaram and Pruhs~\cite{KalPru93}. If the weights are arbitrary, then the competitive ratio cannot be bounded. To address this, both papers considered the restriction where the edge weights are distances in some metric on the set of vertices. They give a $2n-1$ competitive algorithm, which is the best-possible for deterministic algorithms. When randomized algorithms are allowed, the best-known competitive ratio is $O(\log^2(n))$ by Bansal et al.~\cite{BanBucGup07}. If the arrival order is also randomized, then Raghvendra~\cite{raghvendra2016robust} shows that $2\log{(n)}$ is attainable. He also shows that this is the best possible.

% When considering online minimum bipartite matching with non-negative edge weights, perfect matching is a necessity. In in the imperfect minimization setting, an optimal online algorithm simply drops all the requests. In the maximization version, perfect matching makes sense in use-cases where the penalty for a dropped request is sufficiently large that it is effectively not an option.

\subsubsection{$k$-Server Problem}
In the $k$-server problem, $k$ workers are distributed at initial positions in a metric space. Jobs are elements of the same metric space and arrive one at a time. When a job arrives, it must be assigned to a worker which moves to the job's location. The goal in the $k$-server problem is to minimize the total distance traveled by all workers to serve the sequence of jobs. After an assignment, the worker remains available for assignment to new jobs. This \emph{reassignment} distinguishes the $k$-server problem from ours, where workers are fixed to a job once assigned.
	
The $k$-server problem was introduced by Manasse, McGeoch, and Sleater~\cite{ManMcgSle90}. A review of the $k$-server problem literature was written by Koutsoupias~\cite{Kou09}. For randomized algorithms in discrete metrics, the competitive ratio $O(\log^2{(k)} \log{(n)})$ was attained by Bubeck et. al.~\cite{bubeck2017server}, where $n$ is the number of points in the discrete metric space. On the other hand, $\Omega(\log{(k)})$ is a known lower bound. In the i.i.d.~setting, Dehghani et. al. \cite{dehghani2017stochastic} consider a different kind of competitive ratio: they give an online algorithm with a cost no worse than $O(\log{(n)})$ times the cost of the optimal \emph{online} algorithm.

\begin{table}[t]
	\centering
	\caption{Best-known competitive ratios and impossibility bounds for various online bipartite matching problems. $\bigstar$: Results presented in this paper.}
	\label{tab:literature}
	
	\begin{tabular}{cccccc}
		\toprule
		\textbf{Sense} & \textbf{Matching} & \textbf{Arrivals} & \textbf{Restrictions} & \textbf{Best Known} & \textbf{Best Possible}\\
		\cmidrule(r){1-4} \cmidrule(l){5-6} 
		Max & Imperfect & Advers. & 0/1 & $0.632$~\cite{KarVazVaz90} & $0.632$~\cite{KarVazVaz90}\\
		Max & Imperfect & Rand. Ord. & 0/1 & $0.696$~\cite{mahdian2011online} & $0.823$~\cite{ManGhaSab12}\\
		Max & Imperfect & Rand. Ord. & None & $0.368$~\cite{kesselheim2013optimal} & $0.368$~\cite{kesselheim2013optimal}\\
		Max & Imperfect & i.i.d. & None & $0.730$~\cite{BruSanSri16} & $0.745$~\cite{correa2017posted}\\
		Min & Perfect & Advers. & Metric & $O(\log^2(n))$~\cite{BanBucGup07} & $\Omega(\log(n))$~\cite{MeyNanPop06} \\
        Min & Perfect & Rand. Ord. & Metric & $2 \log{(n)}$~\cite{raghvendra2016robust} & $2 \log{(n)}$~\cite{raghvendra2016robust}\\
        Max & Perfect & Adversarial & None & - & 0~\cite{feldman2009wine} \\
		\cmidrule(r){1-4} \cmidrule(l){5-6}
		\textbf{Max} & \textbf{Perfect} & \textbf{i.i.d.} & \textbf{None} & $\frac{1}{2}^\bigstar$ & $\frac{1}{2}^\bigstar$\\
		\bottomrule
	\end{tabular}
\end{table}

\subsection{Structure of this Work}
This paper is organized as follows. Section \ref{sec:prelims} formally introduces the problem of online perfect bipartite matching with i.i.d.~arrivals and defines the concept of competitive ratio. Section \ref{sec:algorithm} describes \algname{}, presents an example to demonstrate the algorithm, and provides the proof that \algname{} is 0.5-competitive. Section \ref{sec:bound} introduces a family of instances of the online perfect bipartite matching problem for which no online algorithm performs better than $\frac{1}{2}$ in terms of competitive ratio. Finally, section \ref{sec:conclusion} summarizes the results and suggests directions for future research.

\section{Preliminaries}\label{sec:prelims}
The set of workers is denoted by $W$ with size $n = |W|$. The set $J$ denotes the set of job types with size $k = |J|$. For every worker $w \in W$ and job type $j \in J$ there is a utility of $u_{wj} \ge 0$ for assigning a job of type $j$ to worker $w$. Let  $\mathcal{D}(J)$ be a known probability distribution over the job types.

At every time step $t=1, \dots, n$, a single job is drawn i.i.d.~from $J$ according to $\mathcal{D}$. The job must be irrevocably assigned to a worker before the next job arrives. Workers are no longer available after they have been assigned a job. Let $r_j$ denote the expected number of jobs of type $j$ that arrive. After $n$ steps, each worker is assigned to one job and the resulting assignment forms a perfect matching. Our goal is to design a procedure such that the expected sum of the utilities of the resulting perfect matching is as high as possible.

Throughout this work, we will repeatedly use two bipartite graphs; the \emph{expectation graph} $G$ and the \emph{realization graph} $\widehat{G}$. The expectation graph $G = (W, J, E)$ is a complete bipartite graph defined over the set of workers $W$ and the set of job types $J$. An edge $[w,j] \in E$ has associated utility $u_{wj} \ge 0$, for $w \in W$ and $j \in J$. The realization graph $\widehat{G} = (W, \widehat{J}, \widehat{E})$ is the random bipartite graph obtained after all $n$ jobs have arrived. $\widehat{J}$ denotes the set of $n$ jobs that arrived. We use $\hat{j}_t \in \widehat{J}$ to denote the job that arrives at time $t$ and $j_t \in J$ to denote its job type. The edge set $\widehat{E}$ consists of all worker-job pairs, such that $\widehat{G}$ is a complete bipartite graph defined over $W$ and $\widehat{J}$. Every edge $[w, \hat{j}] \in \widehat{E}$ has utility $u_{wj}$, where $j$ is the job type of job $\hat{j}$. It is important to remember that the expectation graph $G$ is deterministic and known in advance whereas the realization graph $\widehat{G}$ is a random graph representing a realization of the job arrival process and is revealed over time.

%See Figure \ref{fig:expectation} for an example of an expectation graph and Figure \ref{fig:realization} for an example of a realization graph. 

An instance of the online perfect bipartite matching problem with i.i.d.~arrivals is defined by the set of workers $W$, the job types $J$, non-negative utilities $u_{wj}$, and a distribution over the job types $\mathcal{D}(J)$. Equivalently, the expectation graph $G$ and the distribution $\mathcal{D}(J)$ defines an instance of this problem. Here we analyze the family of potentially randomized algorithms that return a perfect matching $\hat{M}$ on $\widehat{G}$. The performance of an algorithm $ALG$ for a single realization $\widehat{G}$ is given by:
\begin{equation*}
	ALG(\widehat{G}) = \E{\sum_{[w,j] \in E} u_{wj} I_{wj}},
\end{equation*}
where $I_{wj}$ is a random indicator variable that equals $1$ if $ALG$ assigned a job of type $j$ to worker $w$ and equals $0$ otherwise. For a given problem instance defined by expectation graph $G$ and distribution $\mathcal{D}(J)$, $\E{ ALG(\widehat{G}) }$ measures the algorithm's expected performance over samples of $\widehat{G}$ from $G$ according to $\mathcal{D}(J)$.

The worst-case performance across instances is measured by the \emph{competitive ratio}. Let $OPT(\widehat{G})$ be the maximum weight perfect matching in the realization graph $\widehat{G}$ and let $\E{ OPT(\widehat{G}) }$ be its expectation across different realizations for a given expectation graph $G$ and distribution $\mathcal{D}(J)$. $\E{ OPT(\widehat{G}) }$ measures the performance of an optimal algorithm that has full information about the arrival sequence. This is known as an adaptive online adversary. The ratio 
$\frac{\E{ ALG(\widehat{G}) }}{\E{ OPT(\widehat{G}) }}$
measures the performance of $ALG$ relative to the optimal algorithm for a given instance of the problem. The competitive ratio is the worst-case, i.e. lowest, ratio among all possible instances of the expectation graph $G$ and distributions $\mathcal{D}(J)$:

\begin{definition}[Competitive Ratio]
	An algorithm $ALG$ is said to have a competitive ratio of $\alpha$ when for all instances of the expectation graph $G$ and distribution $\mathcal{D}(J)$:
	\begin{equation*}
		\alpha \le \frac{\E{ ALG(\widehat{G}) }}{\E{ OPT(\widehat{G}) }}.
	\end{equation*}
\end{definition}

\subsection{Bounding the Performance of OPT} \label{sec:opt_bound}
It is difficult to compute $\E{ OPT(\widehat{G}) }$ directly. We show that the randomness in $\widehat{G}$ reduces the expected value of the optimal perfect matching compared to the value of the optimal transportation problem where the number of jobs of each type is equal to its expectation. This offline transportation problem is then used to guide the online assignment.

A similar approach was used in the context of unweighted online imperfect bipartite matching by Feldman et al.~\cite{feldman2009online} and Haepler et al.~\cite{HaeMirZad11}. Here, we use a transportation problem instead of a maximum weight matching. We also bound the performance of OPT differently.

Recall that, in expectation, $r_j$ jobs of job type $j \in J$ will arrive in $\hat{G}$. An optimal fractional matching of these jobs is obtained by solving a fractional transportation problem on the expectation graph $G$, where each job type has a demand of $r_j$ and each worker has a supply of $1$ and the sum of utilities is maximized.

Formally, let $f_{wj} \geq 0$ be the flow from worker $w \in W$ to job type $j \in J$. This can be interpreted as a fractional assignment of worker $w$ to jobs of job type $j$. We define the transportation problem $TPP$:
\begin{alignat*}{3}
	TPP(G) = \max_{f_{wj} \ge 0} && \quad \sum_{w \in W} \sum_{j \in J} u_{wj} f_{wj}, & \\
	&& \sum_{w \in W} f_{wj} & = r_j \quad \forall j \in J, \\
	&& \sum_{j \in J} f_{wj} & = 1  \quad \forall w \in W.
\end{alignat*}
Let $f^*_{wj}$ be an optimal flow on edge $[w, j] \in E$.

We claim that $\E{ OPT(\widehat{G}) } \le TPP(G)$. The reason is that the weighted average of perfect matchings $OPT(\widehat{G})$ forms a feasible solution to the transportation problem above.

\begin{lemma}\label{lem:opt_better_expected}
	Given any expectation graph $G$ and distribution over job types $\mathcal{D}(J)$,
	\begin{equation*}
		\E{ OPT(\widehat{G}) } \le TPP(G).
	\end{equation*}
\end{lemma}

\begin{proof}
	Assign each edge in $G$ an indicator variable $I_{wj}$, which takes on the value $1$ if $OPT$ assigns worker $w$ to a job of type $j$ in $\widehat{G}$ and $0$ otherwise. We claim that $f_{wj} = \E{ I_{wj}}$ forms a feasible solution to the transportation problem in $G$. Indeed, 
	\begin{equation*}
	\sum_{w \in W}\E{ I_{wj}} = \E{ \sum_{w \in J} I_{wj} } = r_j, \qquad\qquad
	\sum_{j \in J}\E{ I_{wj}} = \E{ \sum_{j \in J} I_{wj} } = 1.
	\end{equation*}
	Since $\E{I_{wj}}$ is feasible for the transportation problem, it must have objective smaller than $TPP(G)$:
	\begin{equation*}
	\E{OPT(\widehat{G})} = \E{ \sum_{[w, j] \in E} u_{wj} I_{wj} } = \sum_{[w,j] \in E} u_{ij} \E{I_{wj}} \leq TPP(G).
	\end{equation*}
\end{proof}

% \begin{corollary}
% 	For any expectation graph $G$ and distribution $\mathcal{D}(J)$,
% 	$$
% 	\frac{\E{ALG(\widehat{G})}}{TPP(G)} \leq \frac{ \E{ALG(\widehat{G})}}{\E{OPT(\widehat{G})}}.
% 	$$
% \end{corollary}

This implies that we can bound the performance of an algorithm with respect to $TPP(G)$. We apply this technique in section \ref{sec:ratio_proof}.

\section{A 1/2-Competitive Algorithm}\label{sec:algorithm}

%In this section, we present \algname{}, a $\frac{1}{2}$-competitive algorithm for online perfect bipartite matching with i.i.d.~arrivals. The \algname{} algorithm is introduced in section \ref{sec:paid} and we present an example in section \ref{sec:example}. In section \ref{sec:ratio_proof}, we prove that \algname{} is $\frac{1}{2}$-competitive.

\subsection{The \algname{} Algorithm}\label{sec:paid}

% Upon the arrival of a job the \algname{} algorithm assigns it to an available worker.  After all jobs have arrived, the resulting assignment forms a perfect matching $\hat{M}$ on the realization graph $\widehat{G}$. The expected total utility of $\hat{M}$ is at least half of the expected total utility attained by an optimal algorithm that knows the complete arrival sequence. In other words,  the \algname{} algorithm is $\frac{1}{2}$-competitive for the problem of online perfect bipartite matching.

% \algname{} uses an offline transportation problem on the expectation graph $G$ to guide the assignment of arriving jobs to workers.

Before any jobs arrive, \algname{} solves the offline transportation problem $TPP$ on the expectation graph $G$. We find an optimal flow $f^*_{wj}$ from workers to jobs. Throughout the online stage, the algorithm reconstruct this flow between job types and workers as much as possible. For each arriving job, a \emph{preferred worker} $w^P$ is randomly selected with a probability proportional to the optimal flow $f^*$ between the corresponding job type and the worker in the transportation problem. If the preferred worker is no longer available, then the job is assigned to a worker selected randomly from the set of available workers $AW$. We refer to this worker as the \emph{assigned worker} $w^A$. The resulting assignment forms a perfect matching on $\widehat{G}$ since each worker is assigned at most once and each job is assigned to a worker. 

In the context of online bipartite matching, the idea of using an offline solution to guide the online algorithm was used in the ``Suggested Matching'' algorithm~\cite{feldman2009online} and subsequent work, e.g.~\cite{HaeMirZad11}. Our algorithm differs in two ways. First, the offline solution is a transportation problem instead of a maximum weight matching problem. Second, the job is randomly assigned instead of discarded when the preferred worker is no longer available. This random selection ensures that we obtain a perfect matching and is crucial for lemma~\ref{lem:equal_assignment}. The analysis of the competitive performance of \algname{} is also novel except for lemma~\ref{lem:draw_worker}.

The algorithm is formally defined in Algorithm \ref{alg:suggested-matching}. We prove the following result:
\begin{theorem}\label{thm:half_comp}
	\algname{} achieves a competitive ratio of $\frac{1}{2}$ for the online perfect bipartite matching problem with i.i.d.~arrivals. 
\end{theorem}

%Note that a worker is selected for every incoming job. This remains feasible throughout the algorithm since the number of arriving jobs equals the number of available workers and all worker-job combinations are feasible. The \algname{} algorithm will thus construct a perfect bipartite matching which is a feasible solution for the online perfect bipartite matching problem.

\begin{algorithm}
	\caption{\algname{}}
	\label{alg:suggested-matching}
	\begin{algorithmic}
		\STATE \textbf{Input: } Expectation graph $G$.
		\STATE \textbf{Output: } Perfect matching $\hat{M}$ on $\widehat{G}$.
		\STATE 
		\STATE \textbf{Initialization: }
		\STATE Solve the transportation problem $TTP$ on $G$ to obtain the optimal flow $f^*$.
		\STATE $\hat{M} \leftarrow \emptyset$
		\STATE $AW \leftarrow W$
		\STATE
		\STATE \textbf{Online stage:}
		\FOR{$t = 1, \dots, n$}
        	\STATE \texttt{\# Job $\hat{j}_t$ arrives with job type $j_t$.}
			\STATE Randomly draw \textbf{preferred} worker $w^P$ with probability $p(w) = \frac{f^*_{wj_t}}{r_{j_t}}$ for $w \in W$.
			\STATE \texttt{\# Use preferred worker ($w^P)$ as assigned worker ($w^A$) if possible.}
			\IF {$w^P \in AW$}
				\STATE $w^A \leftarrow w^P$
			\ELSE
				\STATE Randomly draw $w^A \in AW$ with equal probability.
			\ENDIF
			\STATE $\hat{M} \leftarrow \hat{M} \cup [w^A, \hat{j}_t]$
			\STATE $AW \leftarrow AW \setminus \{w^A\}$
		\ENDFOR
	\end{algorithmic}
\end{algorithm}

\subsection{Example}\label{sec:example}

To illustrate \algname{}, we consider the example shown in Figure \ref{fig:example}. The example has five workers ($n=5$) and three job types ($k=3$).
The expectation graph is shown in Figure \ref{fig:expectation}. Note that the distribution over job types, $\mathcal{D}(J)$, is fully specified by $r_j$. An instance of the realization graph is shown in Figure \ref{fig:realization}. 

Figure \ref{fig:flow} shows $f^*$, the solution to the transportation problem on $G$ that is used by \algname{}. The corresponding objective value is $TPP(G)=8$. Figures \ref{fig:step1} to \ref{fig:blocked} show the arrival of the jobs and the corresponding assignment made by \algname{}. Figure \ref{fig:blocked} illustrates an instance where the preferred worker selected by \algname{} is not available, and a different worker is assigned. For this particular realization $\widehat{G}$, the perfect matching constructed by \algname{} has a total utility $6$, while the optimal perfect matching on $\widehat{G}$ has a total utility $8$. Note that these values are for this particular realization of $\widehat{G}$. The performance guarantee is with respect to the expectation over all realizations of $\widehat{G}$.

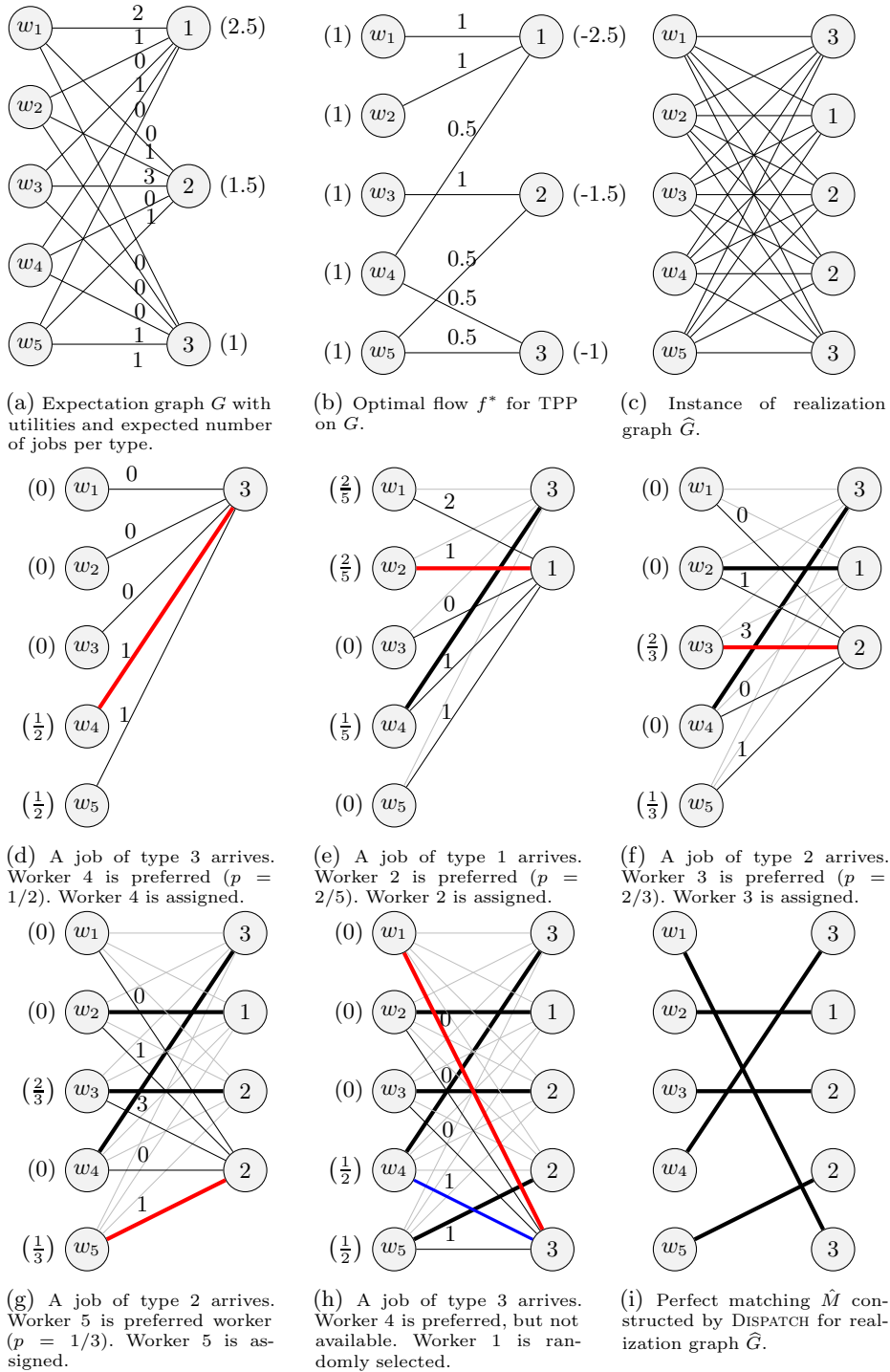
\begin{figure}[p]
	\footnotesize
	\begin{subfigure}[t]{0.3\textwidth}
		\centering
		\begin{tikzpicture}[x=0.55cm, y=0.55cm] 
		\node[node] at(0,6) (worker 1) {$w_1$};
		\node[node] at(0,4) (worker 2) {$w_2$};
		\node[node] at(0,2) (worker 3) {$w_3$};
		\node[node] at(0,0) (worker 4) {$w_4$};
		\node[node] at(0,-2) (worker 5) {$w_5$};
		\node[node] at(4,6) (type 1) {$1$};
		\node[node] at(4,2) (type 2) {$2$};
		\node[node] at(4,-2) (type 3) {$3$};
		
		%		\draw (worker 1.west) node[left] {(1)};
		%		\draw (worker 2.west) node[left] {(1)};
		%		\draw (worker 3.west) node[left] {(1)};
		%		\draw (worker 4.west) node[left] {(1)};
		%		\draw (worker 5.west) node[left] {(1)};
		\draw (type 1.east) node[right] {(2.5)};
		\draw (type 2.east) node[right] {(1.5)};
		\draw (type 3.east) node[right] {(1)};
		
		\draw (worker 1) edge node[pos=0.75,above]{2} (type 1);
		\draw (worker 2) edge node[pos=0.75,above]{1} (type 1);
		\draw (worker 3) edge node[pos=0.74,above]{0} (type 1);
		\draw (worker 4) edge node[pos=0.73,above]{1} (type 1);
		\draw (worker 5) edge node[pos=0.72,above]{0} (type 1);
		
		\draw (worker 1) edge node[pos=0.83,above]{0} (type 2);
		\draw (worker 2) edge node[pos=0.84,above]{1} (type 2);
		\draw (worker 3) edge node[pos=0.85,above, yshift=-0.08cm]{3} (type 2);
		\draw (worker 4) edge node[pos=0.84,yshift=0.1cm]{0} (type 2);
		\draw (worker 5) edge node[pos=0.83,yshift=0.1cm]{1} (type 2);
		
		\draw (worker 1) edge node[pos=0.72,below]{0} (type 3);    
		\draw (worker 2) edge node[pos=0.73,below]{0} (type 3);
		\draw (worker 3) edge node[pos=0.74,below]{0} (type 3);  
		\draw (worker 4) edge node[pos=0.75,below]{1} (type 3);    
		\draw (worker 5) edge node[pos=0.75,below]{1} (type 3);
		
		\end{tikzpicture}
		\caption{\scriptsize Expectation graph $G$ with utilities and expected number of jobs per type.}
		\label{fig:expectation}
	\end{subfigure}
	\hfill
	\begin{subfigure}[t]{0.3\textwidth}
		\centering
		\begin{tikzpicture}[x=0.55cm, y=0.55cm] 
		\node[node] at(0,6) (worker 1) {$w_1$};
		\node[node] at(0,4) (worker 2) {$w_2$};
		\node[node] at(0,2) (worker 3) {$w_3$};
		\node[node] at(0,0) (worker 4) {$w_4$};
		\node[node] at(0,-2) (worker 5) {$w_5$};
		\node[node] at(4,6) (type 1) {$1$};
		\node[node] at(4,2) (type 2) {$2$};
		\node[node] at(4,-2) (type 3) {$3$};
		
		\draw (worker 1.west) node[left] {(1)};
		\draw (worker 2.west) node[left] {(1)};
		\draw (worker 3.west) node[left] {(1)};
		\draw (worker 4.west) node[left] {(1)};
		\draw (worker 5.west) node[left] {(1)};
		\draw (type 1.east) node[right] {(-2.5)};
		\draw (type 2.east) node[right] {(-1.5)};
		\draw (type 3.east) node[right] {(-1)};

		\draw (worker 1) edge node[above]{1} (type 1);
		\draw (worker 2) edge node[above]{1} (type 1);
		\draw (worker 3) edge node[above]{1} (type 2);
		\draw (worker 4) edge node[pos=0.5,above, yshift=0.15cm]{0.5} (type 1);
		\draw (worker 4) edge node[pos=0.5,above] {0.5} (type 3);
		\draw (worker 5) edge node[pos=0.5,above] {0.5} (type 2);
		\draw (worker 5) edge node[above]{0.5} (type 3);
		
		\end{tikzpicture}
		\caption{\scriptsize Optimal flow $f^*$ for \tpp{} on $G$.}
		\label{fig:flow}
	\end{subfigure}
	\hfill
	\begin{subfigure}[t]{0.3\textwidth}
		\centering
		\begin{tikzpicture}[x=0.55cm, y=0.55cm] 
		\node[node] at(0,6) (worker 1) {$w_1$};
		\node[node] at(0,4) (worker 2) {$w_2$};
		\node[node] at(0,2) (worker 3) {$w_3$};
		\node[node] at(0,0) (worker 4) {$w_4$};
		\node[node] at(0,-2) (worker 5) {$w_5$};
		\node[node] at(4,6) (job 1) {$3$};
		\node[node] at(4,4) (job 2) {$1$};
		\node[node] at(4,2) (job 3) {$2$};
		\node[node] at(4,0) (job 4) {$2$};
		\node[node] at(4,-2) (job 5) {$3$};
		
		\draw (worker 1)--(job 1);
		\draw (worker 2)--(job 1);
		\draw (worker 3)--(job 1);
		\draw (worker 4)--(job 1);
		\draw (worker 5)--(job 1);  
		
		\draw (worker 1)--(job 2);
		\draw (worker 2)--(job 2);
		\draw (worker 3)--(job 2);
		\draw (worker 4)--(job 2);
		\draw (worker 5)--(job 2);  
		
		\draw (worker 1)--(job 3);
		\draw (worker 2)--(job 3);
		\draw (worker 3)--(job 3);
		\draw (worker 4)--(job 3);
		\draw (worker 5)--(job 3);  
		
		\draw (worker 1)--(job 4);
		\draw (worker 2)--(job 4);
		\draw (worker 3)--(job 4);
		\draw (worker 4)--(job 4);
		\draw (worker 5)--(job 4);  
		
		\draw (worker 1)--(job 5);
		\draw (worker 2)--(job 5);
		\draw (worker 3)--(job 5);
		\draw (worker 4)--(job 5);
		\draw (worker 5)--(job 5);     
		\end{tikzpicture}
		\caption{\scriptsize Instance of realization graph $\widehat{G}$.}
		\label{fig:realization}
	\end{subfigure}
	
	\begin{subfigure}[t]{0.3\textwidth}
		\centering
		\begin{tikzpicture}[x=0.55cm, y=0.55cm] 
		\node[node] at(0,6) (worker 1) {$w_1$};
		\node[node] at(0,4) (worker 2) {$w_2$};
		\node[node] at(0,2) (worker 3) {$w_3$};
		\node[node] at(0,0) (worker 4) {$w_4$};
		\node[node] at(0,-2) (worker 5) {$w_5$};
		\node[node] at(4,6) (job 1) {$3$};
		
		\draw (worker 1.west) node[left] {$\left(0\right)$};
		\draw (worker 2.west) node[left] {$\left(0\right)$};
		\draw (worker 3.west) node[left] {$\left(0\right)$};
		\draw (worker 4.west) node[left] {$\left(\frac{1}{2}\right)$};
		\draw (worker 5.west) node[left] {$\left(\frac{1}{2}\right)$};

		\draw (worker 1) edge node[above, pos=0.2]{0} (job 1);  
		\draw (worker 2) edge node[above, pos=0.2]{0} (job 1);
		\draw (worker 3) edge node[above, pos=0.2]{0} (job 1);  
		\draw[red, ultra thick] (worker 4) edge node[black, above, pos=0.2]{1} (job 1);    
		\draw (worker 5) edge node[above, pos=0.2]{1} (job 1);    
		\end{tikzpicture}
		\caption{\scriptsize A job of type 3 arrives. Worker 4 is preferred ($p=1/2$). Worker 4 is assigned.}
		\label{fig:step1}
	\end{subfigure}	
	\hfill
	\begin{subfigure}[t]{0.3\textwidth}
		\centering
		\begin{tikzpicture}[x=0.55cm, y=0.55cm] 
		\node[node] at(0,6) (worker 1) {$w_1$};
		\node[node] at(0,4) (worker 2) {$w_2$};
		\node[node] at(0,2) (worker 3) {$w_3$};
		\node[node] at(0,0) (worker 4) {$w_4$};
		\node[node] at(0,-2) (worker 5) {$w_5$};
		\node[node] at(4,6) (job 1) {$3$};
		\node[node] at(4,4) (job 2) {$1$};
		
		\draw (worker 1.west) node[left] {$\left(\frac{2}{5}\right)$};
		\draw (worker 2.west) node[left] {$\left(\frac{2}{5}\right)$};
		\draw (worker 3.west) node[left] {$\left(0\right)$};
		\draw (worker 4.west) node[left] {$\left(\frac{1}{5}\right)$};
		\draw (worker 5.west) node[left] {$\left(0\right)$};
		
		\draw[very thin, gray!50] (worker 1)--(job 1);  
		\draw[very thin, gray!50] (worker 2)--(job 1); 
		\draw[very thin, gray!50](worker 3)--(job 1); 
		\draw[ultra thick,black] (worker 4)--(job 1); 
		\draw[very thin, gray!50] (worker 5)--(job 1);
		
		\draw (worker 1) edge node[above, pos=0.3]{2} (job 2);
		\draw[red, ultra thick] (worker 2) edge node[above, pos=0.3, black]{1} (job 2);
		\draw (worker 3) edge node[above, pos=0.3]{0} (job 2);
		\draw (worker 4) edge node[above, pos=0.3, yshift=-0.05cm]{1} (job 2);
		\draw (worker 5) edge node[above, pos=0.3]{1} (job 2); 
		\end{tikzpicture}
		\caption{\scriptsize A job of type 1 arrives. Worker 2 is preferred ($p=2/5$). Worker 2 is assigned.}
	\end{subfigure}
	\hfill
	\begin{subfigure}[t]{0.3\textwidth}
		\centering
		\begin{tikzpicture}[x=0.55cm, y=0.55cm] 
		\node[node] at(0,6) (worker 1) {$w_1$};
		\node[node] at(0,4) (worker 2) {$w_2$};
		\node[node] at(0,2) (worker 3) {$w_3$};
		\node[node] at(0,0) (worker 4) {$w_4$};
		\node[node] at(0,-2) (worker 5) {$w_5$};
		\node[node] at(4,6) (job 1) {$3$};
		\node[node] at(4,4) (job 2) {$1$};
		\node[node] at(4,2) (job 3) {$2$};

		\draw (worker 1.west) node[left] {$\left(0\right)$};
		\draw (worker 2.west) node[left] {$\left(0\right)$};
		\draw (worker 3.west) node[left] {$\left(\frac{2}{3}\right)$};
		\draw (worker 4.west) node[left] {$\left(0\right)$};
		\draw (worker 5.west) node[left] {$\left(\frac{1}{3}\right)$};
		
		\draw[very thin, gray!50] (worker 1)--(job 1);  
		\draw[very thin, gray!50] (worker 2)--(job 1); 
		\draw[very thin, gray!50](worker 3)--(job 1); 
		\draw[ultra thick,black] (worker 4)--(job 1); 
		\draw[very thin, gray!50] (worker 5)--(job 1);
		
		\draw[very thin, gray!50] (worker 1)--(job 2);  
		\draw[ultra thick,black] (worker 2)--(job 2); 
		\draw[very thin, gray!50] (worker 3)--(job 2); 
		\draw[very thin, gray!50] (worker 4)--(job 2); 
		\draw[very thin, gray!50] (worker 5)--(job 2); 
		
		\draw (worker 1) edge node[above, pos=0.2]{0} (job 3);
		\draw (worker 2) edge node[above, pos=0.2, yshift=-0.08cm]{1} (job 3);
		\draw[ultra thick,red] (worker 3) edge node[above, pos=0.2, black]{3} (job 3);
		\draw (worker 4) edge node[above, pos=0.2]{0} (job 3);
		\draw (worker 5) edge node[above, pos=0.2]{1} (job 3);     
		\end{tikzpicture}
		\caption{\scriptsize A job of type 2 arrives. Worker 3 is preferred ($p=2/3$). Worker 3 is assigned.}
	\end{subfigure}
	
	\begin{subfigure}[t]{0.3\textwidth}
		\centering
		\begin{tikzpicture}[x=0.55cm, y=0.55cm] 
		\node[node] at(0,6) (worker 1) {$w_1$};
		\node[node] at(0,4) (worker 2) {$w_2$};
		\node[node] at(0,2) (worker 3) {$w_3$};
		\node[node] at(0,0) (worker 4) {$w_4$};
		\node[node] at(0,-2) (worker 5) {$w_5$};
		\node[node] at(4,6) (job 1) {$3$};
		\node[node] at(4,4) (job 2) {$1$};
		\node[node] at(4,2) (job 3) {$2$};
		\node[node] at(4,0) (job 4) {$2$};
		
		\draw (worker 1.west) node[left] {$\left(0\right)$};
		\draw (worker 2.west) node[left] {$\left(0\right)$};
		\draw (worker 3.west) node[left] {$\left(\frac{2}{3}\right)$};
		\draw (worker 4.west) node[left] {$\left(0\right)$};
		\draw (worker 5.west) node[left] {$\left(\frac{1}{3}\right)$};
		
		\draw[very thin, gray!50] (worker 1)--(job 1);  
		\draw[very thin, gray!50] (worker 2)--(job 1); 
		\draw[very thin, gray!50](worker 3)--(job 1); 
		\draw[ultra thick,black] (worker 4)--(job 1); 
		\draw[very thin, gray!50] (worker 5)--(job 1);
		
		\draw[very thin, gray!50] (worker 1)--(job 2);  
		\draw[ultra thick,black] (worker 2)--(job 2); 
		\draw[very thin, gray!50] (worker 3)--(job 2); 
		\draw[very thin, gray!50] (worker 4)--(job 2); 
		\draw[very thin, gray!50] (worker 5)--(job 2);
		
		\draw[very thin, gray!50] (worker 1)--(job 3);  
		\draw[very thin, gray!50] (worker 2)--(job 3); 
		\draw[ultra thick] (worker 3)--(job 3); 
		\draw[very thin, gray!50] (worker 4)--(job 3); 
		\draw[very thin, gray!50] (worker 5)--(job 3); 
		
		\draw (worker 1) edge node[above, pos=0.3]{0} (job 4);
		\draw (worker 2) edge node[above, pos=0.3]{1} (job 4);
		\draw (worker 3) edge node[above, pos=0.3]{3} (job 4);
		\draw (worker 4) edge node[above, pos=0.3]{0} (job 4);
		\draw[ultra thick, red] (worker 5) edge node[above, pos=0.3, black]{1} (job 4);     
		\end{tikzpicture}
		\caption{\scriptsize A job of type 2 arrives. Worker 5 is preferred worker ($p=1/3$). Worker 5 is assigned.}
	\end{subfigure}
	\hfill
	\begin{subfigure}[t]{0.3\textwidth}
		\centering
		\begin{tikzpicture}[x=0.55cm, y=0.55cm] 
		\node[node] at(0,6) (worker 1) {$w_1$};
		\node[node] at(0,4) (worker 2) {$w_2$};
		\node[node] at(0,2) (worker 3) {$w_3$};
		\node[node] at(0,0) (worker 4) {$w_4$};
		\node[node] at(0,-2) (worker 5) {$w_5$};
		\node[node] at(4,6) (job 1) {$3$};
		\node[node] at(4,4) (job 2) {$1$};
		\node[node] at(4,2) (job 3) {$2$};
		\node[node] at(4,0) (job 4) {$2$};
		\node[node] at(4,-2) (job 5) {$3$};
		
		\draw (worker 1.west) node[left] {$\left(0\right)$};
		\draw (worker 2.west) node[left] {$\left(0\right)$};
		\draw (worker 3.west) node[left] {$\left(0\right)$};
		\draw (worker 4.west) node[left] {$\left(\frac{1}{2}\right)$};
		\draw (worker 5.west) node[left] {$\left(\frac{1}{2}\right)$};
		
		\draw[very thin, gray!50] (worker 1)--(job 1);  
		\draw[very thin, gray!50] (worker 2)--(job 1); 
		\draw[very thin, gray!50](worker 3)--(job 1); 
		\draw[ultra thick,black] (worker 4)--(job 1); 
		\draw[very thin, gray!50] (worker 5)--(job 1);
		
		\draw[very thin, gray!50] (worker 1)--(job 2);  
		\draw[ultra thick,black] (worker 2)--(job 2); 
		\draw[very thin, gray!50] (worker 3)--(job 2); 
		\draw[very thin, gray!50] (worker 4)--(job 2); 
		\draw[very thin, gray!50] (worker 5)--(job 2);
		
		\draw[very thin, gray!50] (worker 1)--(job 3);  
		\draw[very thin, gray!50] (worker 2)--(job 3); 
		\draw[ultra thick] (worker 3)--(job 3); 
		\draw[very thin, gray!50] (worker 4)--(job 3); 
		\draw[very thin, gray!50] (worker 5)--(job 3); 
		
		\draw[very thin, gray!50] (worker 1)--(job 4);  
		\draw[very thin, gray!50] (worker 2)--(job 4); 
		\draw[ultra thick, black] (worker 5)--(job 4); 
		\draw[very thin, gray!50] (worker 4)--(job 4); 
		\draw[very thin, gray!50] (worker 3)--(job 4);
		
		\draw[ultra thick, red] (worker 1) edge node[above, pos=0.3, black]{0} (job 5);
		\draw (worker 2) edge node[above, pos=0.3]{0} (job 5);
		\draw (worker 3) edge node[above, pos=0.3]{0} (job 5);
		\draw[blue, very thick] (worker 4) edge node[above, pos=0.3, black]{1} (job 5);
		\draw (worker 5) edge node[above, pos=0.3]{1} (job 5);     
		\end{tikzpicture}
		\caption{\scriptsize A job of type 3 arrives. Worker 4 is preferred, but not available. Worker 1 is randomly selected.}
		\label{fig:blocked}
	\end{subfigure}
	\hfill
	\begin{subfigure}[t]{0.3\textwidth}
		\centering
		\begin{tikzpicture}[x=0.55cm, y=0.55cm] 
		\node[node] at(0,6) (worker 1) {$w_1$};
		\node[node] at(0,4) (worker 2) {$w_2$};
		\node[node] at(0,2) (worker 3) {$w_3$};
		\node[node] at(0,0) (worker 4) {$w_4$};
		\node[node] at(0,-2) (worker 5) {$w_5$};
		\node[node] at(4,6) (job 1) {$3$};
		\node[node] at(4,4) (job 2) {$1$};
		\node[node] at(4,2) (job 3) {$2$};
		\node[node] at(4,0) (job 4) {$2$};
		\node[node] at(4,-2) (job 5) {$3$};
		
		\draw[ultra thick, black] (worker 4)--(job 1);
		\draw[ultra thick, black] (worker 2)--(job 2); 
		\draw[ultra thick, black] (worker 3)--(job 3); 
		\draw[ultra thick, black] (worker 5)--(job 4); 
		\draw[ultra thick, black] (worker 1)--(job 5);
		\end{tikzpicture}
		\caption{\scriptsize Perfect matching $\hat{M}$ constructed by \algname{} for realization graph $\widehat{G}$.}
		\label{fig:matching}
	\end{subfigure}
	
	\caption{\small An example of the \algname{} algorithm on the realization graph shown in Figure \ref{fig:realization}. The underlying expectation graph $G$ with $n=5$ and $k=3$ is shown in Figure \ref{fig:expectation}. In Figures \ref{fig:step1} up to \ref{fig:blocked}, the numbers in parenthesis denote the probability of selecting that worker as the preferred worker. Red edges represent the assignment made by the algorithm,  thick black edges are previous assignments, and blue edges mark unavailable preferred workers. Figure \ref{fig:blocked} shows an instance where the preferred worker is busy.}
	\label{fig:example}
\end{figure}

\subsection{Proof of $\frac{1}{2}$-competitiveness}\label{sec:ratio_proof}
To prove that the perfect matching produced by \algname{} has a competitive ratio of a $\frac{1}{2}$, we rely on a key feature of \algname{}: It maintains the invariant, lemma \ref{lem:equal_worker}, that workers are equally likely to be available even though the distribution over job types may not be uniform. To prove this invariant, we first show that both the preferred and the assigned worker are selected uniformly across workers. Recall that the preferred worker may be different than the assigned worker. In fact, the preferred worker does not have to be available and could have been assigned to another job already. Lemma \ref{lem:draw_worker} states this formally for the selection of the preferred worker. The observation underlying this lemma is that each worker is selected with a probability proportional to the total flow $f^*$ originating at the worker, which is equal to one for each worker.

Throughout this section we use additional notation. Let the random variable $W_t^P$ represent the preferred worker for the job arriving at time $t$, and let the random variable $W_t^A$ be the assigned worker. Furthermore, let the random set $AW_t$ consist of the available workers when the job at time $t$ arrives. 
We make no further assumptions on the expectation graph $G$ and/or distribution $\mathcal{D}(J)$ other than those outlined in section \ref{sec:prelims}. Lemmas and theorems in this section are therefore applicable to all problem instances.  

\begin{lemma}\label{lem:draw_worker}
	At each time $t$, the \textbf{preferred} worker $W^P_t$ is drawn uniformly from all workers:
	\begin{equation*}
		\P{W^P_t = w} = \frac{1}{n} \quad \text{for all } w \in W \text{ and } t = 1,\dots, n.
	\end{equation*}
\end{lemma}
\begin{proof}
	By conditioning on the job type $j_t$ at stage $t$ and using the law of total probability, we can rewrite the probability of selecting worker $w$ as:
	\begin{align*}
		\P{W^P_t = w} &= \sum_{j \in J} \P{W^P_t = w | j_t = j} \P{j_t = j}.	\intertext{Since the jobs are drawn i.i.d., a job of type $j$ is selected with probability $ \P{j_t = j} = \frac{r_j}{n}$, by definition of $r_j$. Given a job of type $j$, the algorithm selects a worker $w$ as the preferred worker with probability $\P{W^P_t = w | j_t = j} = \frac{f^*_{wj}}{r_j}$. Thus,}
		\P{W^P_t = w} 
		&= \sum_{j \in J} \frac{f^*_{wj}}{r_j} \frac{r_j}{n} 
		= \sum_{j \in J} \frac{f^*_{wj}}{n}.
	\end{align*}
	Finally, recall that every worker supplies a unit of flow in the offline transportation problem, equivalent to the expected number of jobs it serves. The edges adjacent to worker $w$ must thus transport a unit of flow, so $\sum_{j} f^*_{wj} = 1$. Thus, $\P{W^P_t = w} = \frac{1}{n}$.
\end{proof}

Next we show that the assigned worker is selected uniformly at random from the set of available workers. For this lemma to hold, it is crucial that the draw of the assigned worker is done uniformly at random when the preferred worker is not available. Recall that $W_t^A$ is the assigned worker for the job arriving at time $t$ and that $AW_t$ are the available workers before the job arrives.

\begin{lemma}\label{lem:equal_assignment}
At each time step $t$, the \textbf{assigned} worker $W^P_t$ is drawn uniformly from the available workers:
\begin{equation*}
	\P{W^A_t = w |w \in AW_t} = \frac{1}{n-(t-1)}.
\end{equation*}
\end{lemma}
\begin{proof}
Assume that $w$ is fixed and that $w \in AW_t$. There are two ways for $w$ to be the assigned worker. Either $w$ is the preferred worker or the preferred worker is not available and $w$ is randomly selected. We express this as:
\begin{align*}
\P{W^A_t = w | w \in AW_t} &= \P{W^P_t = w | w \in AW_t} \\ 
    & \quad + \P{W^A_t = w | W^P_t \notin AW_t, w \in AW_t} \times \\
    & \quad \quad \, \P{W^P_t \notin AW_t | w \in AW_t}
\intertext{The selection of $W^P_t$ is independent of whether $w \in AW_t$. Therefore,}
\P{W^A_t = w | w \in AW_t} &= \P{W^P_t = w } \\
    & \quad + \P{W^A_t = w | W^P_t \notin AW_t, w \in AW_t} \P{W^P_t \notin AW_t}
\intertext{Now we use three observations to complete the proof. First, lemma \ref{lem:draw_worker} implies that $\P{W^P_t = w} = \frac{1}{n}$. Second, since there are $t-1$ busy workers, lemma \ref{lem:draw_worker} implies that $\P{W^P_t \notin AW_t} = \frac{(t-1)}{n}$.
Third, the fact that the assigned worker is drawn uniformly at random when the preferred worker is not available implies that $\P{W^A_t = w | W^P_t \notin AW_t, w \in AW_t} = \frac{1}{n-(t-1)}$. Thus,}
\P{W^A_t = w | w \in AW_t} &= \frac{1}{n} + \frac{1}{n-(t-1)} \frac{(t-1)}{n} = \frac{1}{n-(t-1)}.
\end{align*}
\end{proof}

Lemma \ref{lem:equal_assignment} specifies each available worker is equally likely to be assigned to the next job. As a consequence, we can derive the probability that a worker is still available after $t-1$ jobs have arrived:

\begin{lemma}\label{lem:equal_worker}
	\algname{} maintains the following invariant throughout the online stage:
	\begin{equation*}
		\P{w \in AW_t} = \frac{n-(t-1)}{n} \quad \text{for all } w \in W \text{ and } t = 1,\dots, n.
	\end{equation*}
\end{lemma}
\begin{proof}
At every time step, a worker is chosen randomly from the remaining available workers, as shown in lemma \ref{lem:equal_assignment}. The probability that an available worker in time step $t$ is still available in time step $t+1$ is:
\begin{align*}
	\P{w \in AW_{t+1} | w \in AW_t} 
	    &= 1 - \P{W^A_t = w | w \in AW_t} \\ 
	    &= 1 - \frac{1}{n-(t-1)} = \frac{n-t}{n-(t-1)}.
\end{align*}
Thus, the probability of being available for the $t$\textsuperscript{th} job is equal to: 
\begin{align*}
	\P{w \in AW_t} 
		&= \prod_{i=1}^t{\P{w \in AW_{t} | w \in AW_{t-1}}}\\
		&= \frac{n- (t-1)}{n - (t-2)} \frac{n- (t-2)}{n - (t-3)} \dots \frac{n-1}{n} = \frac{n-(t-1)}{n}.
\end{align*}
\end{proof}

From lemma \ref{lem:equal_worker}, we know the probability that a worker is available at each time step. We use this to bound the probability that a worker $w$ is assigned to a job with job type $j$ by \algname{}. We use the indicator random variable $I_{wj}$. $I_{wj} = 1$ when the \algname{} assigns worker $w$ to a job with job type $j$, and $I_{wj} = 0$ otherwise. We bound the probability with respect to $f^*_{wj}$ in $TPP(G)$. By bounding the algorithm's performance with respect to $TPP(G)$ we can bound the competitive ratio of \algname{}. See section \ref{sec:opt_bound} for more details.

\begin{lemma}\label{lem:edge_in_matching}
	Given a perfect matching $\hat{M}$ constructed by \algname{}, the probability that worker $w$ is assigned to a job of type $j$ is bounded by:
	\begin{equation*}
		\P{I_{wj} = 1} \ge \frac{1}{2} f^*_{wj}.
	\end{equation*}
\end{lemma}
\begin{proof}
	If $I_{wj} = 1$, then worker $w$ must have been assigned to a job of type $j$ in one of the time steps. Thus, $I_{wj} = \sum_{t=1}^n I^t_{wj}$ where $I^t_{wj}$ is indicator for whether worker $w$ is assigned to a job of type $j$ at time step $t$:
	\begin{equation*}
		\P{I_{wj} = 1} = \sum_{t=1}^n \P{I^t_{wj} = 1}.
	\end{equation*}
	Let us bound the probability $\P{I^t_{wj} = 1}$ for all $t=1, \dots, n$. First, we condition on the job type arriving at time $t$. Note that $j_t$ must equal $j$:
	\begin{equation*}
		\P{I^t_{wj} = 1} = \P{I^t_{wj} = 1 | j_t = j} \P{j_t = j}.
    \end{equation*}
	Recall that there are two ways for worker $w$ to be assigned after a job of type $j$ arrives. Either $w$ is the preferred worker and is assigned the job, or another worker $w'$ is selected as the preferred worker but is not available. $w$ is then selected as the assigned worker. We lower bound the probability that worker $w$ is assigned for the job of type $j$ by considering only the case where $w$ is the preferred worker.
    \begin{align*}
		\P{I^t_{wj} = 1}
			&\ge \P{w \in AW_{t}, W^P_t = w | j_t = j} \P{j_t = j}\\
			&= \P{w \in AW_t} \P{W^P_t = w| j_t = j}  \P{j_t = j}\\
			&= \frac{n-(t-1)}{n} \frac{f^*_{wj}}{r_j} \frac{r_j}{n}\\
			&= \frac{1}{n}\frac{n-(t-1)}{n} f^*_{wj}.
	\end{align*}
	For the first equality, we use that the job type at time $t$ and the selection of the preferred worker are independent from whether $w$ is available at time $t$.
	The second equality follows from lemma \ref{lem:equal_worker}, the weighted random selection of the preferred worker, and the job arrival process.
	
	We use $\P{I^t_{wj} = 1} = \frac{1}{n}\frac{n-(t-1)}{n} f^*_{wj}$ to bound the total probability of assigning worker $w$ for a job of type $j$:
	$$
	\P{I_{wj} = 1}
		= \sum_{t=1}^n \P{I^t_{wj} = 1} \ge \sum_{t=1}^n \frac{1}{n}\frac{n-(t-1)}{n} f^*_{wj}
		= \frac{1}{2} \frac{n+1}{n}f^*_{wj}
		\ge \frac{1}{2}f^*_{wj}.
	$$
\end{proof}

Lemma \ref{lem:edge_in_matching} bounds the probability that worker $w$ is matched to a job of type $j$. By linearity of expectation, Theorem \ref{thm:half_comp} and the $\frac{1}{2}$ competitive ratio follow almost immediately from lemma \ref{lem:edge_in_matching}.

\begin{proof}[Proof of Theorem \ref{thm:half_comp}]
	The expected utility returned by the algorithm is a weighted sum of indicators whether worker $w$ is assigned to a job of type $j$. Note that each worker is assigned to at most one job (type). We can then apply lemma \ref{lem:edge_in_matching} to bound the probability $P(I_{wj} = 1)$ and the expected utility of the algorithm:
	\begin{align*}
		\E{\algname{}(\widehat{G})}
			&= \E{ \sum_{w \in W, j \in J} u_{wj} I_{wj} } \\
			&= \sum_{w \in W, j \in J} u_{wj} \E{I_{wj}} \\
			&= \sum_{w \in W, j \in J} u_{wj} \P{I_{wj} = 1} \\
		&\ge  \frac{1}{2} \sum_{w \in W, j \in J}  u_{wj} f^*_{wj}= \frac{1}{2} TPP(G).
	\end{align*}
	Note that the inequality requires that the utility weights are non-negative.  
	
	Finally, we apply Lemma \ref{lem:opt_better_expected} to obtain a bound on the competitive ratio attained by \algname{} for any expectation graph $G$ and distribution $\mathcal{D}(J)$:
	\begin{equation*}
		\E{\algname{}(\widehat{G})} \ge \frac{1}{2} TPP(G) \ge \frac{1}{2} \E{OPT(\widehat{G})}
	\end{equation*}
\end{proof}

\section{Best-Possible Competitive Ratio}\label{sec:bound}
We present here a family of instances for which any online algorithm attains a competitive ratio of at most $\frac{1}{2}$. The \algname{} algorithm guarantees a competitive ratio of $\frac{1}{2}$ and is thus optimal with respect to competitive ratio.

\begin{theorem}\label{thm:bounds}
	For the online perfect bipartite matching problem with an i.i.d.~arrival process, no online algorithm can achieve a competitive ratio $\frac{\E{ALG(\widehat{G})}}{\E{OPT(\widehat{G})}}$ better than $\frac{1}{2}$.
\end{theorem}

\begin{proof}
	% We construct an instance of the problem where both statements hold. We first derive $TPP(G)$ and $\E{OPT(\widehat{G})}$ and then we derive $\E{ALG(\widehat{G})}$.
	
	Consider an instance $G$ with the number of job types $k=n+1$. Let the job types be indexed from $1$ to $n+1$ and the workers from $1$ to $n$. Job types $1$ to $n$ each arrive with probability $p/n$ and job type $n+1$ arrives with probability $1-p$. For this graph, we set $u_{wj}=1$ if $w = j$ and to 0 otherwise. This implies $u_{w,n+1} = 0$ for all $w \in W$.
	
	Note that $OPT$ gains a utility of one per unique job type in $\{1, \ldots, n\}$ that arrives. The expected number of unique job types is computed by considering each job type as a geometric random variable with a success probability of $\frac{p}{n}$. Thus, $\E{OPT(\widehat{G})} = n \left(1 - \left(1 - \frac{p}{n}\right)^n\right)$.
		
	For any online algorithm $\text{ALG}^*$, $t-1$ workers are no longer available at time step $t$ regardless of the strategy. Thus, with probability $(1-p) + p \frac{t-1}{n}$ the increase in utility is zero. Thus, the total expected utility increases by at most $p \frac{n - (t-1)}{n}$ in time step $t$. The total expected utility obtained by $ALG^*$ is then:
	\begin{equation*}
		\E{ALG^*(\widehat{G})} \le p \frac{n}{n} + p \frac{n-1}{n} + p \frac{n-2}{n} + \dots + p \frac{1}{n} = \frac{1}{2} p (n+1)
	\end{equation*}
	We compute the relevant ratio and then take the limit as $n$ goes to infinity:
	\begin{equation*}
		\lim_{n \rightarrow \infty} \frac{\E{ALG^*(\widehat{G})}}{\E{OPT(\widehat{G})}} 
        = \lim_{n \rightarrow \infty} \frac{\frac{1}{2} p (n+1)}{n \left(1 - \left(1 - \frac{p}{n}\right)^n\right)} 
        = \frac{1/2 \cdot p}{1 - e^{-p}}
	\end{equation*}
    Since $p$ can take on any value in the interval $(0, 1)$, we consider the limit as $p$ goes to zero:
    \begin{equation*}
    	\lim_{p \rightarrow 0^+} \frac{1/2 \cdot p}{1 - e^{-p}}
        = \lim_{p \rightarrow 0^+} \frac{1/2}{e^{-p}}
        = \frac{1}{2}.
    \end{equation*}
\end{proof}

\begin{corollary}
	\algname{} achieves the best-possible competitive ratio of $\frac{1}{2}$ for the Online Perfect Bipartite Matching problem.
\end{corollary}

\section{Conclusion}\label{sec:conclusion}

In this paper, we examine the problem of online perfect bipartite matching with i.i.d.~arrivals from a known distribution. We present the \algname{} algorithm. It attains a competitive ratio of $\frac{1}{2}$. We show that this is the best possible. Thus, the algorithm \algname{} is optimal in terms of competitive ratio. 

There is an intriguing difference between online perfect bipartite matching  algorithms for minimization and the \algname{} algorithm for maximization. Whereas the competitive ratio for minimization is bounded logarithmically, a constant bound was obtained for maximization with i.i.d.~arrivals. This raises the question of whether a constant competitive ratio is possible for minimization with i.i.d.~arrivals. 

It may be possible to translate the analysis in this work to other contexts. Our analysis relied on two key ideas; the use of the expectation graph and proving that, regardless of how the jobs arrive, the \algname{} algorithm effectively translates the non-uniform sampling over jobs to a uniform sampling over workers. 

\emergencystretch=1em % fix bad boxes in references?

\bibliographystyle{splncs04}
\bibliography{references}

\begin{thebibliography}{10}
\providecommand{\url}[1]{\texttt{#1}}
\providecommand{\urlprefix}{URL }
\providecommand{\doi}[1]{https://doi.org/#1}

\bibitem{BanBucGup07}
Bansal, N., Buchbinder, N., Gupta, A., Naor, J.S.: An {$O(\log^2
  k)$}-competitive algorithm for metric bipartite matching. In: European
  Symposium on Algorithms. pp. 522--533. Springer (2007).
  \doi{10.1007/978-3-540-75520-3\_47}

\bibitem{BruSanSri16}
Brubach, B., Sankararaman, K.A., Srinivasan, A., Xu, P.: New algorithms, better
  bounds, and a novel model for online stochastic matching. In: 24th Annual
  European Symposium on Algorithms. vol.~57, pp. 24:1--24:16. Schloss
  Dagstuhl--Leibniz-Zentrum fuer Informatik (2016).
  \doi{10.4230/LIPIcs.ESA.2016.24}

\bibitem{bubeck2017server}
Bubeck, S., Cohen, M.B., Lee, Y.T., Lee, J.R., Madry, A.: K-server via
  multiscale entropic regularization. In: Proceedings of the 50th Annual ACM
  SIGACT Symposium on Theory of Computing. pp. 3--16. ACM (2018).
  \doi{10.1145/3188745.3188798}

\bibitem{correa2017posted}
Correa, J., Foncea, P., Hoeksma, R., Oosterwijk, T., Vredeveld, T.: Posted
  price mechanisms for a random stream of customers. In: Proceedings of the
  2017 ACM Conference on Economics and Computation. pp. 169--186. ACM (2017).
  \doi{10.1145/3033274.3085137}

\bibitem{dehghani2017stochastic}
Dehghani, S., Ehsani, S., Hajiaghayi, M., Liaghat, V., Seddighin, S.:
  Stochastic k-server: How should uber work? In: 44th International Colloquium
  on Automata, Languages, and Programming. vol.~80, pp. 126:1--126:14. Schloss
  Dagstuhl--Leibniz-Zentrum fuer Informatik, Dagstuhl, Germany (2017)

\bibitem{feldman2009wine}
Feldman, J., Korula, N., Mirrokni, V., Muthukrishnan, S., P{\'a}l, M.: Online
  ad assignment with free disposal. In: International Workshop on Internet and
  Network Economics. pp. 374--385. Springer (2009).
  \doi{10.1007/978-3-642-10841-9\_34}

\bibitem{feldman2009online}
Feldman, J., Mehta, A., Mirrokni, V., Muthukrishnan, S.: Online stochastic
  matching: Beating 1-1/e. In: 50th Annual IEEE Symposium on Foundations of
  Computer Science. pp. 117--126. IEEE (2009). \doi{10.1109/FOCS.2009.72}

\bibitem{HaeMirZad11}
Haeupler, B., Mirrokni, V.S., Zadimoghaddam, M.: Online stochastic weighted
  matching: Improved approximation algorithms. In: International Workshop on
  Internet and Network Economics. pp. 170--181. Springer (2011).
  \doi{10.1007/978-3-642-25510-6\_15}

\bibitem{KalPru93}
Kalyanasundaram, B., Pruhs, K.: Online weighted matching. Journal of Algorithms
   \textbf{14}(3),  478--488 (1993). \doi{10.1006/jagm.1993.1026}

\bibitem{KarVazVaz90}
Karp, R.M., Vazirani, U.V., Vazirani, V.V.: An optimal algorithm for on-line
  bipartite matching. In: Proceedings of the 22nd Annual ACM Symposium on
  Theory of Computing. pp. 352--358. ACM (1990). \doi{10.1145/100216.100262}

\bibitem{kesselheim2013optimal}
Kesselheim, T., Radke, K., T{\"o}nnis, A., V{\"o}cking, B.: An optimal online
  algorithm for weighted bipartite matching and extensions to combinatorial
  auctions. In: European Symposium on Algorithms. pp. 589--600. Springer
  (2013). \doi{10.1007/978-3-642-40450-4\_50}

\bibitem{KhuMitVaz94}
Khuller, S., Mitchell, S.G., Vazirani, V.V.: On-line algorithms for weighted
  bipartite matching and stable marriages. Theoretical Computer Science
  \textbf{127}(2),  255--267 (1994). \doi{10.1016/0304-3975(94)90042-6}

\bibitem{Kou09}
Koutsoupias, E.: The k-server problem. Computer Science Review  \textbf{3}(2),
  105--118 (2009). \doi{10.1016/j.cosrev.2009.04.002}

\bibitem{mahdian2011online}
Mahdian, M., Yan, Q.: Online bipartite matching with random arrivals: an
  approach based on strongly factor-revealing lps. In: Proceedings of the 43rd
  annual ACM symposium on Theory of computing. pp. 597--606. ACM (2011).
  \doi{10.1145/1993636.1993716}

\bibitem{ManMcgSle90}
Manasse, M.S., McGeoch, L.A., Sleator, D.D.: Competitive algorithms for server
  problems. Journal of Algorithms  \textbf{11}(2),  208--230 (1990).
  \doi{10.1016/0196-6774(90)90003-W}

\bibitem{ManGhaSab12}
Manshadi, V.H., Gharan, S.O., Saberi, A.: Online stochastic matching: Online
  actions based on offline statistics. Mathematics of Operations Research
  \textbf{37}(4),  559--573 (2012). \doi{10.1287/moor.1120.0551}

\bibitem{Meh13}
Mehta, A., et~al.: Online matching and ad allocation. Foundations and Trends in
  Theoretical Computer Science  \textbf{8}(4),  265--368 (2013).
  \doi{10.1561/0400000057}

\bibitem{MeyNanPop06}
Meyerson, A., Nanavati, A., Poplawski, L.: Randomized online algorithms for
  minimum metric bipartite matching. In: Proceedings of the 17th annual
  ACM-SIAM Symposium on Discrete Algorithms. pp. 954--959. Society for
  Industrial and Applied Mathematics (2006)

\bibitem{raghvendra2016robust}
Raghvendra, S.: A robust and optimal online algorithm for minimum metric
  bipartite matching. In: Approximation, Randomization, and Combinatorial
  Optimization. Algorithms and Techniques. vol.~60, pp. 18:1--18:16. Schloss
  Dagstuhl--Leibniz-Zentrum fuer Informatik (2016).
  \doi{10.4230/LIPIcs.APPROX-RANDOM.2016.18}

\end{thebibliography}

\end{document}